\documentclass{acmtrans2m}
\pdfoutput=1
\usepackage[utf8x]{inputenc}
\usepackage{amsmath,amssymb}
\usepackage[vlined,ruled,linesnumbered]{algorithm2e}
\usepackage{todonotes}

\usepackage{xspace}
\usepackage{graphicx}
\usepackage{color}

\newcommand{\field}[1]{\mathbb{#1}}

\newcommand{\GFZ}{\ensuremath{\field{F}_2}\xspace}
\newcommand{\ord}[1]{\ensuremath{\mathcal{O}\!\left(#1\right)}}
\newcommand{\kbar}{\ensuremath{\overline{k}}\xspace}
\newcommand{\rbar}{\ensuremath{\overline{r}}\xspace}

\newcommand{\Magma}{\textsc{Magma}\xspace}

\newcommand{\xBG}{\textsf{x86\_64}\xspace}
\newcommand{\Opteron}{\textsf{Opteron}\xspace}
\newcommand{\CTD}{\textsf{Core 2 Duo}\xspace}
\newcommand{\CT}{\textsf{Core 2}\xspace}
\newcommand{\Xeon}{\textsf{Xeon}\xspace}

\newcommand{\PLEI}{Algorithm~\ref{alg:ple_iterative}\xspace}

\newcommand{\addrowsfromtable}{\textsc{AddRowsFromTable}}
\newcommand{\maketable}{\textsc{Make\-Table}}
\newcommand{\gausssubmatrix}{\textsc{Gauss\-Submatrix}}
\newcommand{\ple}{\textsc{Ple}\xspace}
\newcommand{\partialple}{\textsc{PartialPle}\xspace}

\SetKwInput{KwEnsures}{Ensures}
\SetKwInput{KwAssumes}{Assumes}

\newtheorem{definition}{Definition}
\newtheorem{lemma}{Lemma}

\title{Efficient Dense Gaussian Elimination over the Finite Field with Two Elements}
\author{MARTIN R.\ ALBRECHT\\INRIA, Paris-Rocquencourt Center, SALSA Project
UPMC Univ Paris 06, UMR 7606, LIP6, F-75005, Paris, France; CNRS, UMR 7606, LIP6, F-75005, Paris, France\\ GREGORY V.\ BARD\\ University of Wisconsin-Stout, Menomonie, Wisconsin 54751, USA \and CL{\'E}MENT PERNET\\INRIA-MOAIS LIG, Grenoble Univ. ENSIMAG, Antenne de Montbonnot 51, avenue Jean Kuntzmann, F-38330 MONTBONNOT SAINT MARTIN, France
}


\begin{document}

\begin{abstract}
In this work we describe an efficient implementation of a hierarchy of algorithms for Gaussian elimination upon dense matrices over the field with two elements (\GFZ). We discuss both well-known and new algorithms as well as our implementations in the M4RI library, which has been adopted into SAGE. The focus of our discussion is a block iterative algorithm for PLE decomposition which is inspired by the M4RI algorithm. The implementation presented in this work provides considerable performance gains in practice when compared to the previously fastest implementation. We provide performance figures on \xBG CPUs to demonstrate the alacrity of our approach.
\end{abstract}

\maketitle

\section{Introduction}
We describe an efficient implementation of a hierarchy of algorithms for Gaussian elimination of dense matrices over the field with two elements (\GFZ). In particular, we describe algorithms for PLE decomposition \cite{jeannerod-pernet-storjohann:ple2010} which decompose a matrix $A$ into $P\cdot L \cdot E$ where $P$ is a permutation matrix, $L$ is a unit lower triangular matrix and $E$ a matrix in row echelon form. Since matrix decomposition is an essential building block for solving dense systems of linear and non-linear equations \cite{lazard:eurocal83} much research has been devoted to improving the asymptotic complexity of such algorithms. A line of research which produced decompositions such as PLUQ~\cite{GoVa96}, LQUP~\cite{IbMoHu82}, LSP~\cite{IbMoHu82} and PLE~\cite{jeannerod-pernet-storjohann:ple2010}, among others. Each of those decompositions can be reduced to matrix-matrix multiplication and hence can be achieved in  $\ord{n^\omega}$ time where $\omega$ is the complexity exponent of linear algebra\footnote{For practical purposes we set $\omega = 2.807$. Lower exponents have been obtained in theory \cite{CW90}.}. However, in \cite{jeannerod-pernet-storjohann:ple2010} it was shown that many of these decompositions are essentially equivalent (one can be produced from the other in negligible cost) and that for almost all applications the PLE decomposition is at least as efficient in terms of time and memory as any of the other. Hence, in this work, we focus on the PLE decomposition but consider the special case of \GFZ. In particular, we propose a new algorithm for block-iterative PLE decomposition and discuss its implementation. We also describe our implementation of previously known algorithms in the M4RI library \cite{M4RI}. 

This article. is organised as follows. We will first highlight some computer architecture issues and discuss in place methods, as well as define notation in the remainder of this section. We will start by giving the definitions of reduced row echelon forms (RREF) and PLE decomposition in Section~\ref{sec:definitions}. We will then discuss the three algorithms and their implementation issues for performing PLE decomposition in Section~\ref{sec:algorithms}, and other implementation issues in Section~\ref{sec:implementation}. We conclude by giving empirical evidence of the viability of our approach in Section~\ref{sec:results}. In Appendix~\ref{app:m4ri} we will describe the original M4RI algorithm (for reducing \GFZ matrices into either row echelon form or RREF), which initiated this line of research.

\subsection{Computer Architecture Issues}

The M4RI library implements dense linear algebra over $\GFZ$. Our implementation focuses on 64-bit x86 architectures (\xBG), specifically the Intel \CT and the AMD \Opteron. Thus, we assume in this work that native CPU words have 64 bits. However, it should be noted that our code also runs on 32-bit CPUs and on non-\textsf{x86} CPUs such as the PowerPC.

Element-wise operations over $\GFZ$, being mathematically trivial, are relatively cheap compared to memory access. In fact, in this work we demonstrate that the two fastest implementations for dense matrix decomposition over $\GFZ$ (the one presented in this work and the one found in \Magma~\cite{magma} due to Allan Steel) perform worse for moderately sparse matrices despite the fact that fewer field operations are performed. Thus our work provides further evidence that more refined models for estimating the expected efficiency of linear algebra algorithms on modern CPUs and their memory architectures are needed.

An example of the CPU-specific issues would be the ``SSE2 instruction set,'' which we will have cause to mention in Section~\ref{say_sse2}. This instruction set has a command to XOR (exclusive-OR) two 128-bit vectors, with a single instruction. Thus if we are adding a pair of rows of $128n$ bits, we can do that with $n$ instructions. This is a tremendous improvement on the $128n$
floating-point operations required to do that with real-numbered problems. This also means that any algorithm which touches individual bits instead of blocks of 128 bits will run about two orders of magnitude slower than an algorithm which avoids this.

\subsection{Notations and space efficient representations} \label{sec:notation}
By $A[i,j]$ we denote the entry of the matrix $A$ at row $i$ and column $j$. By $\ell_i$ we denote the $i$th element of the vector $\ell$. We represent $m \times m$ permutation matrices as integer vectors of length $m$, i.e., in LAPACK-style. That is to say, the permutation matrix 
$$\left[ \begin{array}{cccccc}
1 & 0 & 0 & 0 & 0 \\
0 & 0 & 1 & 0 & 0 \\
0 & 1 & 0 & 0 & 0 \\
0 & 0 & 0 & 0 & 1 \\
0 & 0 & 0 & 1 & 0 
\end{array} \right]$$
is stored as $P=[0, 2, 2, 4, 4]$, where for each index $i$ the entry $P_i$ encodes which swap ought to be performed on the input. This representation allows to apply permutations in-place.

All indices start counting at zero, i.e. $A[2,2]$ refers to the intersection of the third row and third column. All algorithms, lemmas, theorems, propositions in this work are presented for the special case of \GFZ. Hence, any statement that some algorithm produces a certain output is always understood to apply to the special case of \GFZ,
unless stated otherwise.

Matrix triangular decompositions allow that the lower triangular and the upper triangular matrices can be stored one above the other in the same amount of memory as for the input matrix: their non trivial coefficients occupy disjoint areas, as the diagonal of one of them has to be unit and can therefore be omitted. For two $m\times n$ matrices $L$ and $U$ such that $L=[\ell_{i,j}]$ is unit lower triangular and $U = [u_{i,j}]$ is upper triangular, we shall use the notation $L\setminus U$, following~\cite{jeannerod-pernet-storjohann:ple2010} to express the fact that $L$ and $U$ are stored one next to the other within the same $m\times n$ matrix. Thus, $A = L\setminus U$ is the $m \times n$ matrix $A = [a_{i,j}]$ such that  $a_{i,j} = \ell_{i,j}$ if $i>j$, and $a_{i,j} = u_{i,j}$ otherwise. We call it a space sharing representation of $L$ and $U$. Furthermore, $A = A_0 \mid A_1$ means that $A$ is obtained by augmenting $A_0$ by $A_1$.

The methods of this paper not only produce space-sharing triangular decompositions, but also overwrite their input matrix with the output triangular decomposition. However, they are not in-place, contrarily to the general algorithms in~\cite{jeannerod-pernet-storjohann:ple2010}, as they use tabulated multiplications that require non constant extra memory allocations.

Last but not least, complexity is always reported for square matrices, but the complexities for over-defined and under-defined matrices can be derived with moderate effort.

\section{Matrix Decompositions}
\label{sec:definitions}

In this work we are interested in matrix decompositions that do not destroy the column profile, since those play a crucial role in linear algebra for Gr\"obner basis computations \cite{f4}\footnote{These  are also useful in simplified variants of the F4 algorithm for polynomial system solving such as the popular \cite{courtois-klimov-patarin-shamir:eurocrypt2000}.}.

\begin{definition}[Column rank profile] The column rank profile of a matrix $A$ is the lexicographically smallest sequence of column indices $i_0 \leq \dots < i_{r}$ , with $r$ the rank of $A$ such that the corresponding columns are linearly independent.
\end{definition}

\begin{definition}[Leading Entry] The leading entry in a non-zero row of a matrix is the leftmost non-zero entry in that row.\end{definition}

Perhaps, the most well-known of such decompositions is the row echleon form.

\begin{definition}[Row Echelon Form] An $m \times n$ matrix $A$ is in row echelon form if any all-zero rows are grouped together at the bottom of the matrix, and if the leading entry of each non-zero row is located to the right of the leading entry of the previous row.
\end{definition}

While row echelon forms are not unique, \emph{reduced} row echelon forms, whose definition we reproduce below, are.

\begin{definition}[Reduced Row Echelon Form]
An $m \times n$ matrix $A$ is in reduced row echelon form if it is in row echelon form, and furthermore, each leading entry of a non-zero row is the only non-zero element in its column.
\end{definition}

In order to compute (reduced) row echelon forms, we compute the $PLE$ decomposition.

\begin{definition}[PLE Decomposition]
Let $A$ be an $m\times n$ matrix over a field $\field{F}$. A PLE decomposition
of $A$ is a triple of matrices $P,L$ and $E$ such that $P$ is an $m\times m$
permutation matrix, $L$ is an $m\times m$ unit lower triangular matrix, and $E$ is an $m\times n$ matrix in row-echelon form, and $A=PLE$.
\end{definition}

We refer the reader to \cite{jeannerod-pernet-storjohann:ple2010} for a proof that any matrix has a PLE decomposition. There is also proven that the PLE decomposition is at least as efficient for almost all applications as other column rank profile revealing decompositions in terms of space and time complexity.

One last technicality stems from the fact that a computer's memory is addressed in a one-dimensional fashion, while matrices are two-dimensional objects. The methods of this work use row major representations, not column major representations. The definition is as follows:

\begin{definition}[Row Major Representation]
A computer program stores a matrix in row major representation if the entire first row is stored in memory, followed by the entire second row, followed by the
entire third row, and so forth.
\end{definition}

\section{Algorithms}
\label{sec:algorithms}

The PLE decomposition can be computed in sub-cubic time complexity by a block recursive algorithm reducing most algebraic operations to matrix multiplication and some other much cheaper operations \cite{jeannerod-pernet-storjohann:ple2010}.
We denote this algorithm as ``\emph{block recursive} PLE decomposition'' in this work. In principle this recursion could continue until $1\times1$ matrices are reached as a base case. However, in order to achieve optimal performance, these recursive calls must at some dimension ``cross over'' to a base case implementation, presumably much larger than $1\times 1$. This base case has  asymptotic complexity worse than the calling algorithm but better performance, in practice, for small dimensions. 

In this work, we also focus on optimising this base case.  This is done in two phases: a block iterative algorithm using tabulated updates (based on the ``Method of four Russians'' inversion algorithm \cite{Ba06}) denoted by \textsc{BlockIterativePLE} on top of a lazy iterative partial decomposition  which we denote \textsc{PartialPLE}. Here, lazy has a technical meaning which we
define as follows. For some entry $A[i,j]$ below the main diagonal, i.e. with $i>j$, we do not modify row $i$ in such a way as to make $A[i,j]=0$ until the last possible moment. As it comes to pass, that moment will be when we check to see if $A[i,i]=0$, to see if it would be a suitable pivot element.

\subsection{Partial iterative algorithm}
\label{sec:algorithms-classical}

This algorithm is used as a base case of the block iterative algorithm of section~\ref{sec:algorithms-iterative}. It is based on the classical $\ord{n^3}$ Gaussian elimination algorithm, where pivots are found by searching column-wise: the algorithm only moves to the next column when all rows positions of the current column have been found to be zero. This condition ensures that the $E$
matrix will be in row echelon form and the column indices of the leading entries of each row (i.e. the pivots) will indicate the column rank profile of the input matrix.

To improve cache efficiency, the updates are done in a ``lazy'' manner, inspired by the left looking elimination strategy: updates are delayed to the moment when the algorithm actually needs to check whether a coefficient is non-zero. At that moment, all trailing coefficients of the current row are also updated, as they are in the same cache page. This process requires the maintenance of a table $s=[s_0,\dots,s_{r-1}]$ storing for each pivot column the highest-numbered row considered so far. It is defined by $s_i=\max\{i_j,j=0\dots i\}$, where $i_j$ is the row number where the $j$th column's pivot has been found. Consequently, the sequence $s_i$ is monotonically decreasing. As a result, the algorithm produces the PLE decomposition of only $s=\max\{i_j,j=0\dots i\}$ rows of the input matrix. We therefore call it a ``lazy'' iterative partial PLE decomposition.  The complete algorithm is shown as Algorithm~\ref{alg:ple_gauss} and the key idea is illustrated in Figure~\ref{fig:iterativeple}. 

\begin{algorithm}
\KwIn{$A$ -- an $m \times n$ matrix}
\KwResult{Returns $r$, the rank of $A$, a permutation vector $P$ of size $s$, a vector
  $Q$ of size $r$ and an integer $s$ with $r\leq s \leq m$}
\KwEnsures{$A \leftarrow 
  \begin{bmatrix}
    L\setminus E\\
    A_1
  \end{bmatrix}$ where $L$ is $s\times r$ unit lower triangular, $E$ is an $r\times
  n $ echelon form, $A=  \begin{bmatrix}A_0\\A_1  \end{bmatrix}$, $PLE=A_0$ and $\text{rank}(A_0)=r$.}

\SetKw{KwAnd}{and}
\SetKw{KwBreak}{break}

\Begin{
    $r, c,  \leftarrow 0,0$\;
  \tcp{$n$ dim.\ array, s.t.\ column $i$
      has been updated to row $s_i$.}
    $s\longleftarrow [0,\dots, 0]$\; 
  \While{$r<m$ \KwAnd $c < n$} {
    $\text{found} \longleftarrow \texttt{False}$\;
    \For(\tcp*[h]{search for some pivot}){$c \leq j < n$}{
      \For{$r \leq i < m$}{
        \If{$r>0$ \KwAnd $i > s_{r-1}$}{
          Let $k\leq r-1$ be such that $s_{k-1}\geq i > s_{k}$\;
          \For(\tcp*[h]{Update row $i$ from column $j$ to $n$}){$k \leq l < r$}{
            \lIf{$A[i,Q_l]\neq 0$}{ add the row $l$ to the row $i$ starting at
              column $j$\;}
            $s_l\longleftarrow i$\;
          }
        }
        \lIf{$A[i,j]\neq 0$}{
          $\text{found} \leftarrow \texttt{True}$ \KwAnd
          \KwBreak;
        }
      }
      \lIf{\text{found}}{\KwBreak;}
    }
    \eIf{\text{found}}{
      $P_r, Q_r \longleftarrow i,j$\;
      swap the rows $r$ and $i$\;      
      $r, c \longleftarrow r+1,j+1$\;
    }{$j\longleftarrow j+1$\;}
  }
  $s_{max}\longleftarrow \max(s_i)$\;
  \For{$0\leq j < r$}{
    \For{$s_j+1\leq i \leq s_{max}$}{
      \If{$A[i,Q_j]\neq 0$}{ add the row $j$ to the row $i$ starting at
              column $j$\;}
    }
  }
  \tcp{Now compress $L$}
  \lFor{$0 \leq j < r$}{
   swap the columns $j$ and $Q_j$ starting at row $j$\;
  }
  \Return{$P,r,Q,s$};
}
\caption{\partialple: lazy iterative partial PLE Decomposition}
\label{alg:ple_gauss}
\end{algorithm}

\begin{figure}
\begin{center}
  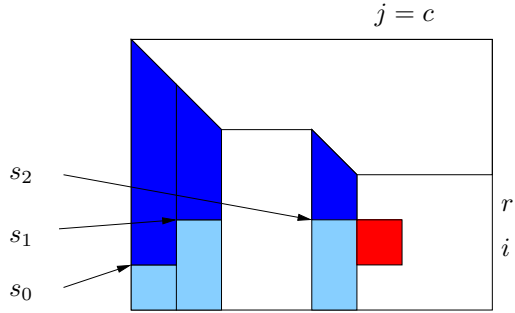
\end{center}
\caption{Lazy PLE elimination: the matrix has been processed through row $r$ and
  column $c$. In order to determine if the coefficient at position $(i,j)$ is
  non zero, the algorithm must apply updates of the two last non zero columns as
  only $s_0$ is greater than $i$ and hence the update of its corresponding
  column was already applied.}
\label{fig:iterativeple}
\end{figure}

As a specialized instance of standard Gaussian elimination, algorithm~\ref{alg:ple_gauss} computes a partial PLE decomposition of an $m\times n$ matrix of rank $r$ in $\ord{mnr}$ operations over $\GFZ$. For the square case, this is obviously $\ord{n^2r}$.

\subsection{Block Iterative}
\label{sec:algorithms-iterative}

For matrices over \GFZ the ``Method of the Four Russians'' Inversion (M4RI), performs Gaussian elimination on a random $n \times n$ matrix in $\ord{n^3/\log n}$. The M4RI algorithm was introduced in \cite{Ba06} and later described in \cite{Ba07} and \cite[Ch. 9]{monograph}. It inherits its name and main idea from ``Konrod's Method for matrix-matrix multiplication,'' which is often mistakenly called ``The Method of Four Russians'' despite the fact that none of the authors of that paper are actually Russians \cite{ArDiKrFa70,AhHoUl74}. For the key idea behind the M4RI algorithm, assume that we have computed the row echelon form of an $k \times n$ submatrix of some matrix $A$ and that this submatrix has rank $k$. Then, in order to eliminate below, we may use a table of all $2^k$ linear combinations of the $k$ pivot rows found in the previous step. Hence, with only one addition per row, $k$ columns are processed. Furthermore, the construction of all $2^k$ linear combinations of the pivots can be accomplished in only $2^k$ vector additions saving a factor of $k$ compared to naive table construction. Finally, if we set $k \approx \log n$ this algorithm can achieve $\ord{n^3/\log n}$ under the assumption that it always finds pivots among the first $c \cdot k$ rows it considers where $c$ is a small constant. If this condition is not met it either fails (see also \cite{Ba07}), or regresses to cubic complexity (this is the strategy is implemented in the M4RI library). We refer the reader to Appendix~\ref{app:m4ri} for a more in-depth treatment of the M4RI algorithm. 

Our implementation of block-iterative PLE decomposition is based on our implementation of the M4RI algorithm and the performance reported in Section~\ref{sec:results} is partly due to techniques inherited from our fast implementation of the M4RI algorithm. Thus, we point out that our implementation of the M4RI algorithm is very efficient, allowing it to beat asymptotically fast implementations of Gaussian elimination over $\GFZ$ up to matrix dimensions of several thousand despite having a worse asymptotic complexity. (See also Appendix~\ref{app:m4ri}.)

\PLEI, which can be seen as the PLE analog of the M4RI algorithm\footnote{Note that \cite{albrecht:phd2010} and earlier versions of this work constain an algorithm denoted MMPF. As it turns out, the MMPF algorithm suffers from the same performance regression as M4RI, in the case when pivots cannot be found among the first $c\cdot k$ considered rows. On the other hand, unlike M4RI, it does compute a PLE decomposition.}, does not require pivots among the first $c\cdot k$ rows in order to achieve $\ord{n^3/\log n}$.  

Informally, the algorithm proceeds as follows. We divide the input matrix $A$ into vertical stripes of $k \approx \log n$ columns each. Then, we perform PLE decomposition on the first such stripe. This operation is cheap because  these stripes are narrow. Now, in order to update the matrix to the right of the current stripe we use the M4RM matrix multiplication algorithm \cite{matmulgf2}, i.e., we compute all possible linear combinations of the pivot rows just discovered and use the rows of the transformation matrix $L$ as a indices in this table. Hence, for each row of $A_{22}$ we peform one row addition from the table instead of $\ord{k}$ row additions. Thus, regardless of the rank of the currently considered stripe, we achieve $\ord{n^3 /\log n}$ because the update on the right is performed using M4RM. 

Note, however, that our presentation of \PLEI is independent of the choice of matrix-multiplication algorithm for clarity of presentation. Thus it could (in principle) be instantiated with naive multiplication, and that would result in a complexity of $\ord{n^3}$.

However, performing PLE decomposition on the current slice has an important disadvantage when compared with the M4RI algorithm: it is less cache friendly. Recall, that our matrices are stored in row major representation. PLE decomposition will always touch all rows when treating the current slice. This is cheap when considering field operations only. However, if the matrix $A$ is sufficient large, this will result in a load of a new cache line for each considered row. To mitigate this effect, we use \partialple as defined in the last section.

Every call to \partialple returns some $\rbar$ and some $s$, which are the dimensions of the $L$ matrix calculated by the
\partialple call. Recall, the dimensions of $A$ are $m\times n$ and $k$ is the parameter of
our algorithm.
Assume that a call to \partialple has returned with $\rbar=k$ and some $s < m$ and hence that we can decompose our matrix as follows
\[A = \left(
\begin{array}{ccccccc}
&A_{00}&&A_{01}&&A_{02}&\\
 &A_{10}&&A_{11} = L_{11} \setminus E_{11} &&\ A_{12} &  \\
 \phantom{blablabla}& A_{20}& \phantom{blablabla}&A_{21} = L_{21}\ & \phantom{blablabla}&A_{22}&\phantom{blablabla}\\
&A_{30}&&A_{31} &\ &A_{32}&\\
\end{array}\right) 
\]
where $A_{11}$ has $r$ rows and $A_{21}$ has $s - r$ rows. We note that if $\rbar < k$ we have $s = m$ and hence $A_{31}$ a matrix with zero rows, i.e., we only need to update $A_{22}$ which reduces to matrix multiplication as discussed above. Hence, we may focus on the $\rbar = k$ case in our discussion.

Algorithm~\ref{alg:ple_iterative} proceeds by first computing $A_{12} \gets (L_{11})^{-1} \cdot A_{12}$. We then construct a table $T$ of all $2^k$ linear combinations of the rows of $A_{12}$. Using Gray codes or similar techniques these tables can be constructed using $2^k$ vector additions only \cite{matmulgf2}. Using this table, we can implement the map $v \rightarrow v \cdot A_{12}$ for all $2^k$ possible vectors $v$ as the table lookup $T(v)$. Furthmore, we can implement the map $v \rightarrow v \cdot (E_{11}^{-1} \cdot A_{12})$ as $(v \cdot E_{11}^{-1}) \cdot A_{12}$. That is, we store a table $T'$ of all $2^k$ possible pairs $(v,v \cdot E_{11}^{-1})$ and compute $v \rightarrow v \cdot E_{11}^{-1} \cdot A_{12}$ by the table lookup $T(T'(v))$.

Finally, we ``postprocess'' $T$ such that we also update $L$ when computing $v \cdot (E_{11})^{-1}$, i.e., we replace any row $v$ of $A_{31}$ by $v \cdot E_{11}^{-1}$ when adding $(v \cdot E_{11}^{-1}) \cdot A_{12}$ to a row of $A_{32}$. For example, assume $\rbar = 3 $ and that the first row of the $\rbar \times n$ submatrix $A_{11}$ has the first $\rbar$ bits equal to \texttt{[1 0 1]} ($=5$). Assume further that we want to clear $\rbar$ bits of a row which also starts with \texttt{[1 0 1]}. Then -- in order to generate $L$ -- we need to encode that this row is cleared by adding the first row only, i.e., we want the first $\rbar = 3$ bits to be \texttt{[1 0 0]}. Thus, to correct the table, we add -- in the postprocessing phase of $T$ -- \texttt{[1 0 0]} to the first $\rbar$ bits $T(5)$. Hence, when adding this row of $T$ to $A_{31}$, we produce \texttt{[1 0 1] $\oplus$ [1 0 1] $\oplus$ [1 0 0] $=$ [1 0 0]} and store it in $A_{31}$.

Using $T$ and $T'$, we can now perform the operations
\begin{itemize}
 \item[] $A_{22} \gets L_{21} \cdot A_{12} + A_{22}$ in line \ref{alg:ple_iterative:nse}.
 \item[] $(A_{31}, A_{32} ) \gets (A_{31} \cdot (E_{11})^{-1}, A_{31} \cdot (E_{11})^{-1} \cdot A_{12} + A_{32} )$ in line~\ref{alg:ple_iterative:sswe}
\end{itemize}
using one vector addition per row. The update of $A_{22}$ is a multiplication by $A_{12}$ plus an addition and is hence accomplished by one lookup to $T$ and one vector addition per row. Similarly, the updates of $A_{31}$ and $A_{32}$ can be achieved using lookups in $T'$ and $T$ and one vector addition per row.

\begin{algorithm}
\KwIn{$A$ -- an $m \times n$ matrix}
\KwResult{Returns $r$, the rank of $A$, a permutation $P$ of size $m$ and an $r$-dimensional vector $Q$} 
\KwEnsures{$A \leftarrow [L \setminus E]$ where $L$ is unit lower triangular and
  $E$ is in echelon form such that $A=PLE$ and $Q_i$ is the column index of the
  $i$th pivot of $E$.}
\SetKw{KwAnd}{and}
\SetKw{KwBreak}{break}

\Begin{
  $r, c \leftarrow 0,0$\;
  $k \longleftarrow $ pick some $0 < k \leq n$; \tcp{$\ord{n^3/\log n} \rightarrow k \approx \log n$}
  $Q\longleftarrow [0,\dots,0]$\;
  \While{$r<m$ \KwAnd $c < n$} {
    \If{$k > n - c$}{$k \leftarrow n - c$\;}
    \tcc{$A=
      \begin{bmatrix}
            A_{00} & A_{01}& A_{02}\\ 
            A_{00} & A_{11}& A_{12}\\ 
      \end{bmatrix}$ where $A_{00}$ is $r\times c$ and $A_{11}$ is $(n-r)\times k$
    }
    $\rbar, \overline{P}, \overline{Q},s \longleftarrow \partialple(A_{11})$\nllabel{alg:ple_iterative:partialple}\;

    \For{$0 \leq i < \rbar$}{
      $Q_{r+i} \longleftarrow c+ \overline{Q}_i$\;
      $P_{r+i} \longleftarrow r+\overline{P}_i$\;
    }

    \tcc{$A=
      \begin{bmatrix}
        A_{00} & A_{01}& A_{02}\\ 
        A_{10} & L_{11} \setminus E_{11}& A_{12}\\ 
        A_{20} & \begin{array}{cc} L_{21}\ \mid &0 \end{array} & A_{22}\\ 
        A_{30} & A_{31} & A_{32}\\ 
      \end{bmatrix}
      $ where $A_{00}$ is $r\times c$ and $L_{11} $ is $\rbar\times \rbar$, $E_{11}$
      is $\rbar \times k$}
    
    $
    \begin{bmatrix}
      A_{10}\\A_{20}
    \end{bmatrix} \longleftarrow 
    P^T    \begin{bmatrix}
      A_{10}\\A_{20}
    \end{bmatrix}$,
    $
    \begin{bmatrix}
      A_{12}\\A_{22}
    \end{bmatrix} \longleftarrow 
    P^T    \begin{bmatrix}
      A_{12}\\A_{22}
    \end{bmatrix}$\;
     $A_{12} \longleftarrow (L_{11})^{-1} \times A_{12}$\nllabel{alg:ple_iterative:ANE}\;
    $A_{22} \gets L_{21} \cdot A_{12} + A_{22}$ \nllabel{alg:ple_iterative:nse} \;
    $(A_{31}, A_{32} ) \gets (A_{31} \cdot E_{11}^{-1}, A_{31} \cdot E_{11}^{-1} \cdot
    A_{12} + A_{32} )$ \nllabel{alg:ple_iterative:sswe}\;
  
    \tcp{Now compress L}
    \For{$0 \leq i < \rbar$}{
      swap the columns $r+i$ and $c+i$ starting at row $r+i$\;
    }
	$r,c \longleftarrow r+\rbar,c+k$\;
  }
  \Return{$P,Q,r$};
}
\caption{Block-iterative \ple Decomposition}
\label{alg:ple_iterative}
\end{algorithm}

Overall we have:
\begin{lemma}
Let $k = \log n$, then \PLEI computes the PLE decomposition of some $n \times n$ matrix $A$ in $\ord{n^3/\log n}$ field operations if the M4RM algorithm is used in line~\ref{alg:ple_iterative:nse} and \ref{alg:ple_iterative:sswe}.
\end{lemma}

\begin{proof}
Let $r_i$ be the rank computed by each call to Algorithm~\ref{alg:ple_gauss}.  Note that the sum of the $r_i$'s is the rank of the matrix, $r$.

The overall operation count is obtained the sum of the following contributions:
\begin{itemize}
\item calls to algorithm~\ref{alg:ple_gauss} (line~\ref{alg:ple_iterative:partialple}):
$$
\sum_{i=0}^{n/k-1}(n-\sum_{j=0}^{i-2}r_j)kr_i \leq nkr,
$$
\item solving with the $r_i\times r_i$ lower triangular matrix $L_{11}$ on
  $A_{12}$ (line~\ref{alg:ple_iterative:ANE}) can be done with the classic  cubic time back substitution algorithm:
$$
\sum_{i=0}^{n/k-1}r_i^2(n-(i+1)k) \leq \sum_{i=0}^{n/k-1}k^2n\leq n^2k,
$$
\item updates based on Gray code tabulated matrix multiplication (line~\ref{alg:ple_iterative:nse}):
$$
\sum_{i=0}^{n/k-1}(n-(i+1)k)(2^k+n-\sum_{j=0}^{i}r_j) \leq  (2^k+n)\sum_{i=1}^{n/k} ik \leq \frac{(2^k+n)n^2}{k},
$$
\item updates in line~\ref{alg:ple_iterative:sswe} share the same table as the
  previous one, line~\ref{alg:ple_iterative:nse}  (all linear combinations of
  rows of $A_{12}$). Hence they only take: 
$$\sum_{i=0}^{n/k-1}(n-(i+1)k)r_i \leq \frac{rn^2}{k}.$$ 
\end{itemize}
For $k=\log n$, the overall complexity is $\ord{n^3/\log n}$.
\end{proof}

\subsection{Block Recursive}
\label{sec:algorithms-recursive}

In \cite{jeannerod-pernet-storjohann:ple2010} an algorithm is given that computes a PLE decomposition in-place and in time complexity $\ord{n^\omega}$ by reduction to matrix-matrix multiplication. This algorithm can be recovered by setting $k=n/2$ in \PLEI and by replacing \partialple with \ple, We give it explicity in Algorithm~\ref{alg:ple_recursive} for clarity of presentation. The asymptotic complexity of the algorithm depends on matrix multiplication algorithm used. Our implementation uses Strassen-Winograd and hence achieves the $\ord{n^\omega}$ running time; this is the subject of a previous paper \cite{matmulgf2}.  We note that $A_{01} \longleftarrow (L_{00})^{-1} \times A_{01}$ is the well-known triangular solving with matrices (TRSM) operation with a lower triangular matrix on the left which can be reduced to matrix-matrix multiplication as well \cite{jeannerod-pernet-storjohann:ple2010}.

\begin{algorithm}[htbp]
\KwIn{$A$ -- an $m \times n$ matrix}
\KwResult{Returns $r$, the rank of $A$, a permutation $P$ of size $m$ and an $r$-dimensional vector $Q$} 
\KwEnsures{$A \leftarrow [L \setminus E]$ where $L$ is unit lower triangular and
  $E$ is in echelon form such that $A=PLE$ and $Q_i$ is the column index of the
  $i$th pivot of $E$.}
\SetKw{KwAnd}{and}
\Begin{
  $n' \longleftarrow$ pick some integer $0 \leq n' < n$; \tcp{$n' \approx n/2$}
  \tcc{$A=
      \begin{bmatrix}
            A_{0} & A_{1}\\ 
      \end{bmatrix}$ where $A_{0}$ is $m\times n'$.
  }

  $r',P',Q' \longleftarrow \ple(A_{0})$; \tcp{first recursive call}

  \For{$0 \leq i \leq r'$}{$Q_i \longleftarrow Q'_i$\;}
  \For{$0 \leq i \leq r'$}{$P_i \longleftarrow P'_i$\;}

  \tcc{$A=
      \begin{bmatrix}
            L_{00} \setminus E_{00} & A_{01}\\ 
            A_{10} & A_{11}\\ 
      \end{bmatrix}$ where $L_{00}$ is $r' \times r'$ and $E_{00}$ is $r' \times n'$.
  }

  \If{$r'$}{
      $A_1 \longleftarrow P \times A_1$\;
      $A_{01} \longleftarrow (L_{00})^{-1} \times A_{01}$\;
      $A_{11} \longleftarrow A_{11} + A_{10} \times A_{01}$\;
  }

   $r'',P'',Q'' \longleftarrow$ \ple($A_{11}$); \tcp{second recursive call}
   
   $A_{10} \longleftarrow P'' \times A_{10}$\;
 
   \tcc{Update P \& Q}
   \For{$0 \leq i < m - r'$}{$P_{r' + i} \longleftarrow P''_i + r'$\;}
   \For{$0 \leq i < n - n'$}{$Q_{n' + i} \longleftarrow Q''_i + n'$\;}

   $j \leftarrow r'$\;
   \For{$n' \leq i < n' + r''$}{$Q_j \longleftarrow Q_i$; $j \leftarrow j + 1$\;}

  \tcc{Now compress L}
  $j \leftarrow n'$\;
  \For{$r' \leq i < r'+r''$}{
   swap the columns $i$ and $j$  starting at row $i$\;
  }
  \Return{$P,Q,r' + r''$}\;
}
\caption{Block-recursive PLE Decomposition}
\label{alg:ple_recursive} 
\end{algorithm}

\section{Implementation details}
\label{sec:implementation}

In our implementation we call Algorithm~\ref{alg:ple_recursive} which recursively calls itself until the matrix fits into L2 cache when it calls \PLEI. Our implementation of \PLEI uses four Gray code tables. (See also \cite{matmulgf2}.)

One of the major bottlenecks are column swaps. Note that this operation is usually not considered when estimating the complexity of algorithms in linear algebra, since it is not a field operation. In Algorithm~\ref{alg:swap} a simple algorithm for swapping two columns $a$ and $b$ is given with bit-level detail. In Algorithm~\ref{alg:swap} we assume that the bit position of $a$ is greater than the bit position of $b$ for clarity of presentation. The advantage of the strategy in Algorithm~\ref{alg:swap} is that it uses no conditional jumps in the inner loop, However, it still requires 9 instructions per row. On the other hand, we can add two rows with $9 \cdot 128 = 1152$ entries in 9 instructions if the ``SSE2 instruction set''\label{say_sse2} is available. Thus, for matrices of size $1152 \times 1152$ it takes roughly the same number of instructions to add two matrices as it does to swap two columns. If we were to swap every column with one other column once during some algorithm it thus would be as expensive as a matrix multiplication (for matrices of these dimensions).

\begin{algorithm}[htbp]
\KwIn{$A$ -- an $m \times n$ matrix}
\KwIn{$a$ -- an integer $0 \leq a < b < n$}
\KwIn{$b$ -- an integer $0 \leq a < b < n$}
\KwResult{Swaps the columns $a$ and $b$ in $A$}
\SetKw{KwAnd}{and}
\Begin{
$M \longleftarrow$  the memory where $A$ is stored\;
$a_w, b_w \longleftarrow$ the word index of $a$ and $b$ in $M$\;
$a_b, b_b \longleftarrow$ the bit index of $a$ and $b$ in $a_w$ and $b_w$\;
$\Delta \longleftarrow$ $a_b - b_b$\;
$a_{m} \longleftarrow$ the bit-mask where only the $a_b$-th bit is set to 1\;
$b_{m} \longleftarrow$ the bit-mask where only the $b_b$-th bit is set to 1\;
\For{$0 \leq i < m$}{
  $R \longleftarrow$ the memory where the row $i$ is stored\;
  $R[a_w] \longleftarrow R[a_w] \oplus ((R[b_w] \odot b_{m}) \gg \Delta)$\;
  $R[b_w] \longleftarrow R[b_w] \oplus ((R[a_w] \odot a_{m}) \ll \Delta)$\;
  $R[a_w] \longleftarrow R[a_w] \oplus ((R[b_w] \odot b_{m}) \gg \Delta)$\;
}
}
\caption{Column Swap}
\label{alg:swap} 
\end{algorithm}

Another bottleneck for relatively sparse matrices in dense ``row-major representation'' is the search for pivots. Searching for a non-zero element in a row can be relatively expensive due to the need to identify the bit position; hence, we cannot work with machine words in bulk but have to read individual bits and see which column they are in. However, the main performance penalty is due to the fact that searching for a non-zero entry in one column in a row-major representation is very cache unfriendly.

Indeed, both our implementation and the implementation available in \Magma suffer from performance regression on relatively sparse matrices as shown in Figure~\ref{fig:sparse-m4ri}. We stress that this is despite the fact that the theoretical complexity of matrix decomposition is rank sensitive. More precisely, strictly fewer field operations must be performed for low-rank matrices. This highlights the distinction between minimizing the running time and minimizing the number of field operations---these are not identical objectives. While the penalty for relatively sparse matrices is much smaller for our implementation than for \Magma, it clearly does not achieve the theoretically optimum performance.

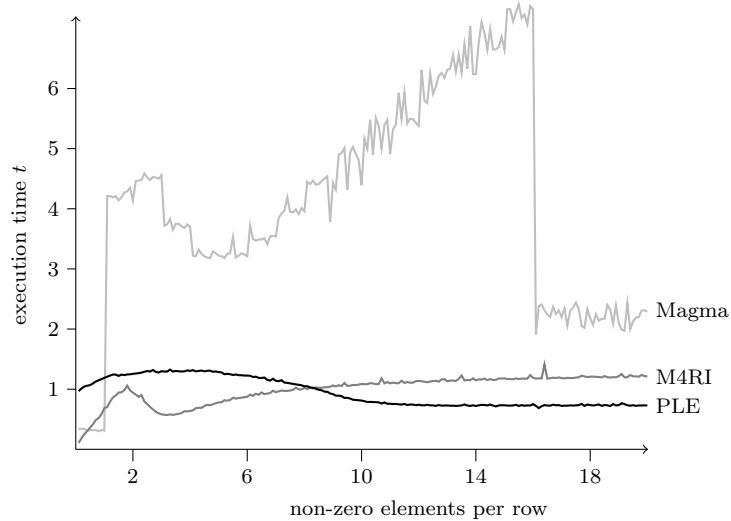
\begin{figure}[htbp]
\begin{center}
\begin{tikzpicture}[xscale=0.038,yscale=0.8]
  \draw[->] (-0.0,-0.0) -- (200.2,-0.0);
  \draw[->] (0,-0.0) -- (0,7.2);
  \draw (-20,3.4) node[rotate=90] {execution time $t$};

  \foreach \y in {1,...,6}
  \draw (0.1,\y) -- (-3,\y) node[anchor=east] {\y};

  \draw (20,-0.0) -- (20,-0.2) node[anchor=north] {2};
  \draw (60,-0.0) -- (60,-0.2) node[anchor=north] {6};
  \draw (100,-0.0) -- (100,-0.2) node[anchor=north] {10};
  \draw (140,-0.0) -- (140,-0.2) node[anchor=north] {14};
  \draw (180,-0.0) -- (180,-0.2) node[anchor=north] {18};

  \draw (120,-1.0) node {non-zero elements per row};

  \draw[thick,color=lightgray] plot[id="magma"] coordinates {(1, 0.34000) (2, 0.34000) (3, 0.34000) (4, 0.31000) (5, 0.32000) (6, 0.32000) (7, 0.32000) (8, 0.30000) (9, 0.32000) (10, 0.31000) (11, 4.21000) (12, 4.21000) (13, 4.19000) (14, 4.21000) (15, 4.14000) (16, 4.19000) (17, 4.27000) (18, 4.29000) (19, 4.35000) (20, 4.14000) (21, 4.46000) (22, 4.47000) (23, 4.48000) (24, 4.59000) (25, 4.52000) (26, 4.48000) (27, 4.53000) (28, 4.52000) (29, 4.50000) (30, 4.56000) (31, 3.72000) (32, 3.74000) (33, 3.83000) (34, 3.65000) (35, 3.75000) (36, 3.75000) (37, 3.71000) (38, 3.68000) (39, 3.74000) (40, 3.70000) (41, 3.22000) (42, 3.21000) (43, 3.32000) (44, 3.25000) (45, 3.21000) (46, 3.19000) (47, 3.18000) (48, 3.29000) (49, 3.26000) (50, 3.22000) (51, 3.21000) (52, 3.18000) (53, 3.25000) (54, 3.26000) (55, 3.53000) (56, 3.19000) (57, 3.22000) (58, 3.26000) (59, 3.24000) (60, 3.21000) (61, 3.72000) (62, 3.50000) (63, 3.47000) (64, 3.49000) (65, 3.49000) (66, 3.51000) (67, 3.41000) (68, 3.54000) (69, 3.55000) (70, 3.54000) (71, 3.85000) (72, 3.92000) (73, 3.97000) (74, 4.19000) (75, 3.95000) (76, 3.94000) (77, 3.99000) (78, 3.91000) (79, 4.01000) (80, 3.96000) (81, 4.45000) (82, 4.41000) (83, 4.47000) (84, 4.40000) (85, 4.41000) (86, 4.46000) (87, 4.50000) (88, 4.54000) (89, 3.78000) (90, 4.43000) (91, 4.32000) (92, 4.90000) (93, 4.93000) (94, 5.01000) (95, 4.32000) (96, 4.94000) (97, 5.03000) (98, 4.91000) (99, 4.81000) (100, 4.39000) (101, 5.16000) (102, 5.00000) (103, 5.48000) (104, 4.90000) (105, 5.51000) (106, 5.38000) (107, 5.01000) (108, 5.39000) (109, 5.48000) (110, 4.91000) (111, 5.32000) (112, 5.40000) (113, 5.93000) (114, 5.50000) (115, 5.93000) (116, 5.42000) (117, 5.50000) (118, 5.49000) (119, 5.43000) (120, 5.38000) (121, 6.31000) (122, 5.80000) (123, 5.76000) (124, 6.25000) (125, 5.91000) (126, 6.03000) (127, 6.21000) (128, 6.26000) (129, 6.18000) (130, 6.31000) (131, 6.33000) (132, 6.56000) (133, 6.28000) (134, 6.56000) (135, 6.27000) (136, 6.73000) (137, 6.33000) (138, 7.04000) (139, 6.24000) (140, 6.24000) (141, 6.77000) (142, 7.10000) (143, 6.92000) (144, 6.65000) (145, 6.99000) (146, 6.65000) (147, 6.65000) (148, 6.69000) (149, 6.86000) (150, 6.65000) (151, 7.32000) (152, 7.32000) (153, 7.12000) (154, 7.26000) (155, 7.40000) (156, 7.17000) (157, 7.23000) (158, 7.07000) (159, 7.37000) (160, 7.32000) (161, 1.91000) (162, 2.38000) (163, 2.41000) (164, 2.31000) (165, 2.24000) (166, 2.20000) (167, 2.37000) (168, 2.25000) (169, 2.31000) (170, 2.20000) (171, 2.34000) (172, 2.06000) (173, 2.30000) (174, 2.36000) (175, 2.44000) (176, 2.33000) (177, 2.05000) (178, 2.02000) (179, 2.33000) (180, 2.18000) (181, 2.07000) (182, 2.25000) (183, 2.14000) (184, 2.17000) (185, 2.17000) (186, 2.40000) (187, 2.32000) (188, 2.09000) (189, 2.39000) (190, 2.10000) (191, 1.99000) (192, 1.97000) (193, 2.43000) (194, 2.00000) (195, 2.13000) (196, 2.19000) (197, 2.20000) (198, 2.31000) (199, 2.31000) (200, 2.29000)}  node[right,color=black] {Magma};

\draw[color=gray,thick] plot[id="m4ri"] coordinates {(1, 0.106155) (2, 0.184164) (3, 0.251415) (4, 0.303108) (5, 0.360643) (6, 0.414627) (7, 0.486798) (8, 0.519742) (9, 0.583870) (10, 0.685364) (11, 0.701443) (12, 0.791403) (13, 0.862657) (14, 0.894982) (15, 0.943845) (16, 0.948925) (17, 0.983653) (18, 1.058876) (19, 0.978244) (20, 0.942851) (21, 0.910335) (22, 0.870607) (23, 0.903626) (24, 0.794898) (25, 0.751917) (26, 0.703680) (27, 0.664177) (28, 0.622128) (29, 0.598736) (30, 0.581246) (31, 0.575591) (32, 0.569314) (33, 0.581121) (34, 0.571415) (35, 0.581529) (36, 0.590549) (37, 0.602096) (38, 0.633154) (39, 0.631056) (40, 0.637458) (41, 0.647052) (42, 0.683402) (43, 0.690109) (44, 0.688369) (45, 0.721489) (46, 0.740237) (47, 0.733376) (48, 0.749635) (49, 0.766576) (50, 0.780792) (51, 0.782709) (52, 0.815303) (53, 0.808925) (54, 0.819675) (55, 0.856786) (56, 0.842952) (57, 0.852394) (58, 0.859198) (59, 0.866153) (60, 0.873672) (61, 0.908846) (62, 0.888298) (63, 0.924500) (64, 0.906349) (65, 0.944918) (66, 0.917566) (67, 0.926394) (68, 0.935242) (69, 0.976125) (70, 0.949351) (71, 0.964972) (72, 0.976726) (73, 0.971648) (74, 0.971378) (75, 0.976877) (76, 0.989472) (77, 1.023465) (78, 0.996700) (79, 0.996345) (80, 1.021938) (81, 1.002980) (82, 1.025406) (83, 1.009642) (84, 1.028584) (85, 1.022169) (86, 1.023683) (87, 1.031212) (88, 1.032265) (89, 1.041096) (90, 1.061411) (91, 1.046542) (92, 1.065558) (93, 1.052070) (94, 1.103617) (95, 1.054645) (96, 1.057528) (97, 1.070088) (98, 1.068171) (99, 1.075558) (100, 1.086820) (101, 1.081535) (102, 1.083108) (103, 1.100390) (104, 1.094615) (105, 1.083225) (106, 1.085532) (107, 1.179791) (108, 1.096924) (109, 1.118066) (110, 1.096977) (111, 1.108275) (112, 1.099314) (113, 1.098875) (114, 1.140278) (115, 1.104692) (116, 1.117734) (117, 1.107452) (118, 1.108677) (119, 1.133074) (120, 1.138049) (121, 1.131893) (122, 1.123428) (123, 1.115298) (124, 1.125756) (125, 1.136285) (126, 1.137816) (127, 1.125501) (128, 1.161839) (129, 1.166295) (130, 1.131556) (131, 1.143585) (132, 1.142175) (133, 1.137177) (134, 1.170231) (135, 1.226686) (136, 1.146127) (137, 1.152069) (138, 1.148524) (139, 1.179800) (140, 1.148448) (141, 1.162584) (142, 1.160667) (143, 1.176126) (144, 1.150629) (145, 1.167743) (146, 1.153145) (147, 1.155713) (148, 1.178170) (149, 1.176623) (150, 1.179899) (151, 1.168482) (152, 1.168863) (153, 1.165564) (154, 1.180034) (155, 1.167097) (156, 1.176009) (157, 1.174169) (158, 1.175008) (159, 1.235052) (160, 1.188697) (161, 1.176353) (162, 1.179927) (163, 1.178066) (164, 1.410760) (165, 1.175002) (166, 1.181250) (167, 1.190936) (168, 1.185536) (169, 1.195529) (170, 1.183506) (171, 1.179883) (172, 1.187214) (173, 1.184755) (174, 1.184619) (175, 1.193059) (176, 1.197349) (177, 1.225413) (178, 1.212423) (179, 1.201041) (180, 1.207087) (181, 1.208996) (182, 1.205363) (183, 1.201408) (184, 1.195647) (185, 1.217675) (186, 1.198349) (187, 1.193616) (188, 1.199563) (189, 1.197402) (190, 1.215841) (191, 1.252375) (192, 1.196108) (193, 1.199229) (194, 1.236302) (195, 1.212404) (196, 1.212873) (197, 1.207745) (198, 1.236736) (199, 1.218035) (200, 1.210933)} node[right,color=black] {M4RI};

\draw[color=black,thick] plot[id="ple"] coordinates {(1, 0.967001) (2, 1.011240) (3, 1.043532) (4, 1.057283) (5, 1.067011) (6, 1.094712) (7, 1.126642) (8, 1.147539) (9, 1.166968) (10, 1.188493) (11, 1.207241) (12, 1.226128) (13, 1.243596) (14, 1.242186) (15, 1.220430) (16, 1.242211) (17, 1.234236) (18, 1.244736) (19, 1.247406) (20, 1.254644) (21, 1.260065) (22, 1.265938) (23, 1.274927) (24, 1.289218) (25, 1.302123) (26, 1.292878) (27, 1.319090) (28, 1.286986) (29, 1.277622) (30, 1.298599) (31, 1.293035) (32, 1.293842) (33, 1.324919) (34, 1.292739) (35, 1.302995) (36, 1.297650) (37, 1.294314) (38, 1.295756) (39, 1.312144) (40, 1.304358) (41, 1.316395) (42, 1.296363) (43, 1.302899) (44, 1.310021) (45, 1.296868) (46, 1.305285) (47, 1.307949) (48, 1.288711) (49, 1.294765) (50, 1.283429) (51, 1.280192) (52, 1.248310) (53, 1.255100) (54, 1.248704) (55, 1.236663) (56, 1.243733) (57, 1.232309) (58, 1.231709) (59, 1.226472) (60, 1.226308) (61, 1.222957) (62, 1.201951) (63, 1.207814) (64, 1.187733) (65, 1.185693) (66, 1.192795) (67, 1.155630) (68, 1.142165) (69, 1.172876) (70, 1.138912) (71, 1.132558) (72, 1.149631) (73, 1.110075) (74, 1.104494) (75, 1.096257) (76, 1.088169) (77, 1.079330) (78, 1.073613) (79, 1.064443) (80, 1.047714) (81, 1.049457) (82, 1.027283) (83, 1.023933) (84, 1.008276) (85, 0.988839) (86, 0.983257) (87, 0.959266) (88, 0.956504) (89, 0.928407) (90, 0.917338) (91, 0.892135) (92, 0.892280) (93, 0.884687) (94, 0.854258) (95, 0.856597) (96, 0.844121) (97, 0.824273) (98, 0.820348) (99, 0.815931) (100, 0.810675) (101, 0.800143) (102, 0.785634) (103, 0.788427) (104, 0.785680) (105, 0.784312) (106, 0.780050) (107, 0.760630) (108, 0.779530) (109, 0.757411) (110, 0.757455) (111, 0.752751) (112, 0.755011) (113, 0.748248) (114, 0.744579) (115, 0.754171) (116, 0.736143) (117, 0.744608) (118, 0.743812) (119, 0.733087) (120, 0.741566) (121, 0.732432) (122, 0.730616) (123, 0.741013) (124, 0.730900) (125, 0.730208) (126, 0.722856) (127, 0.735323) (128, 0.726059) (129, 0.728209) (130, 0.728833) (131, 0.725145) (132, 0.725483) (133, 0.721675) (134, 0.724811) (135, 0.729282) (136, 0.746324) (137, 0.728213) (138, 0.720790) (139, 0.731694) (140, 0.720459) (141, 0.731025) (142, 0.738490) (143, 0.737890) (144, 0.729711) (145, 0.742782) (146, 0.727809) (147, 0.713109) (148, 0.729410) (149, 0.733000) (150, 0.731499) (151, 0.727514) (152, 0.726306) (153, 0.729044) (154, 0.740103) (155, 0.720986) (156, 0.734946) (157, 0.724087) (158, 0.730419) (159, 0.727758) (160, 0.750673) (161, 0.721864) (162, 0.685695) (163, 0.715466) (164, 0.735101) (165, 0.727639) (166, 0.726179) (167, 0.741733) (168, 0.732329) (169, 0.735057) (170, 0.721917) (171, 0.725240) (172, 0.725486) (173, 0.730473) (174, 0.727208) (175, 0.730838) (176, 0.733854) (177, 0.755746) (178, 0.720966) (179, 0.744222) (180, 0.736405) (181, 0.728481) (182, 0.735357) (183, 0.745696) (184, 0.722822) (185, 0.735396) (186, 0.731748) (187, 0.719136) (188, 0.752489) (189, 0.730183) (190, 0.734502) (191, 0.767001) (192, 0.751637) (193, 0.732539) (194, 0.725180) (195, 0.733402) (196, 0.722572) (197, 0.729101) (198, 0.728895) (199, 0.730496) (200, 0.729762)} node[right,color=black] {PLE};

\end{tikzpicture}
\caption{Gaussian elimination of $10,000 \times 10,000$ matrices on Intel 2.33GHz Xeon E5345 comparing Magma 2.17-12 and M4RI 20111004.}
\label{fig:sparse-m4ri}
\end{center}
\end{figure}

\section{Performance}
\label{sec:results}

In Table~\ref{tab:ple-iterative-vs-gauss-c2d}, we consider reduced row echelon forms, i.e., we first compute the PLE decomposition which we then use to recover the reduced row echelon form, of dense random matrices over \GFZ. We compare the running times for Algorithm~\ref{alg:ple_recursive} (with either cubic PLE decomposition or \PLEI as base case).

\begin{table}[htbp]
\begin{footnotesize}
\begin{center}
\begin{tabular}{|c|r|r|}
\hline
Matrix Dimension & {base case: cubic PLE}& {base case: \PLEI}\\
\hline
$10,000 \times 10,000$ & 1.478s &   0.880s \\
$16,384 \times 16,384$ & 5.756s &   3.570s \\
$20,000 \times 20,000$ & 8.712s &   5.678s \\
$32,000 \times 32,000$ &29.705s &  24.240s \\
\hline
\end{tabular}
\caption{Comparing base cases on 64-bit Debian/GNU Linux, 2.33Ghz \CTD.}
\label{tab:ple-iterative-vs-gauss-c2d}
\end{center}
\end{footnotesize}
\end{table}

In Table~\ref{tab:pluq-random}, we give the average running time over ten trials for computing reduced row echelon forms of dense random $n \times n$ matrices over \GFZ. We compare the asymptotically fast implementation due to Allan Steel in \Magma, the cubic Gaussian elimination implemented by Victor Shoup in NTL, and both our implementations. Both the implementation in \Magma and our PLE decomposition reduce matrix decomposition to matrix multiplication. A discussion and comparison of matrix multiplication in the M4RI library and in \Magma can be found in \cite{matmulgf2}. In Table~\ref{tab:pluq-random}, the column `PLE' denotes the complete running time for first computing the PLE decomposition and the computation of the reduced row echelon form from the PLE form. 

\begin{table}[htbp]
\begin{footnotesize}
\begin{center}
\begin{tabular}{|c|r|r|r|r|}
\hline
 & \multicolumn{4}{|c|}{64-bit Linux, 2.6Ghz \Opteron}\\
\hline
$n$             &  {\Magma} &   {NTL} &   {M4RI} &    {PLE}\\
                & {2.15-10} &   5.4.2 & 20090105 & 20100324\\
\hline
$10,000$ &   3.351s &   18.45s &   2.430s &   1.452s\\
$16,384$ &  11.289s &   72.89s &  10.822s &   6.920s\\
$20,000$ &  16.734s &  130.46s &  19.978s &  10.809s\\
$32,000$ &  57.567s &  479.07s &  83.575s &  49.487s\\
$64,000$ & 373.906s & 2747.41s & 537.900s & 273.120s\\
\hline
 & \multicolumn{4}{|c|}{64-bit Linux, 2.33Ghz \Xeon (E5345)}\\
\hline
$n$ & {\Magma} &   {NTL} &   {M4RI} &    {PLE}\\
& {2.16-7} &   5.4.2 & 20100324 & 20100324\\
\hline
$10,000$&   2.660s &  12.05s &   1.360s &   0.864s\\
$16,384$&   8.617s &  54.79s &   5.734s &   3.388s\\
$20,000$&  12.527s & 100.01s &  10.610s &   5.661s\\
$32,000$&  41.770s & 382.52s &  43.042s &  20.967s\\
$64,000$& 250.193s &      -- & 382.263s & 151.314s\\

\hline
\end{tabular}
\caption{Gaussian elimination for random matrices $n \times n$ matrices}
\label{tab:pluq-random}
\end{center}
\end{footnotesize}
\end{table}

\clearpage
\bibliographystyle{acmtrans}
\bibliography{pluqm4ri_biblio}

\begin{thebibliography}{}

\bibitem[\protect\citeauthoryear{Aho, Hopcroft, and Ullman}{Aho
  et~al\mbox{.}}{1974}]{AhHoUl74}
{\sc Aho, A.}, {\sc Hopcroft, J.}, {\sc and} {\sc Ullman, J.} 1974.
\newblock {\em The Design and Analysis of Computer Algorithms.}
\newblock Addison-Wesley.

\bibitem[\protect\citeauthoryear{Albrecht, Bard, and Hart}{Albrecht
  et~al\mbox{.}}{2010}]{matmulgf2}
{\sc Albrecht, M.}, {\sc Bard, G.}, {\sc and} {\sc Hart, W.} 2010.
\newblock Algorithm 898: Efficient multiplication of dense matrices over
  {GF(2)}.
\newblock {\em ACM Trans. Math. Softw.\/}~{\em 37}, 9:1--9:14.

\bibitem[\protect\citeauthoryear{Albrecht}{Albrecht}{2010}]{albrecht:phd2010}
{\sc Albrecht, M.~R.} 2010.
\newblock {A}lgorithmic {A}lgebraic {T}echniques and their application to
  {B}lock {C}ipher {C}ryptanalysis.
\newblock Ph.D. thesis, Royal Holloway, University of London.

\bibitem[\protect\citeauthoryear{Albrecht and Bard}{Albrecht and
  Bard}{2011}]{M4RI}
{\sc Albrecht, M.~R.} {\sc and} {\sc Bard, G.~V.} 2011.
\newblock {\em The {M4RI} Library -- Version 20111004}.
\newblock The M4RI~Team.
\newblock \url{http://m4ri.sagemath.org}.

\bibitem[\protect\citeauthoryear{Arlazarov, Dinic, Kronrod, and
  Faradzev}{Arlazarov et~al\mbox{.}}{1970}]{ArDiKrFa70}
{\sc Arlazarov, V.}, {\sc Dinic, E.}, {\sc Kronrod, M.}, {\sc and} {\sc
  Faradzev, I.} 1970.
\newblock On economical construction of the transitive closure of a directed
  graph.
\newblock {\em Dokl. Akad. Nauk.\/}~{\em 194,\/}~11.
\newblock (in Russian), English Translation in Soviet Math Dokl.

\bibitem[\protect\citeauthoryear{Bard}{Bard}{2006}]{Ba06}
{\sc Bard, G.~V.} 2006.
\newblock Accelerating cryptanalysis with the {M}ethod of {F}our {R}ussians.
\newblock Cryptology ePrint Archive, Report 2006/251.
\newblock Available at \url{http://eprint.iacr.org/2006/251.pdf}.

\bibitem[\protect\citeauthoryear{Bard}{Bard}{2007}]{Ba07}
{\sc Bard, G.~V.} 2007.
\newblock Algorithms for solving linear and polynomial systems of equations
  over finite fields with applications to cryptanalysis.
\newblock Ph.D. thesis, University of Maryland.

\bibitem[\protect\citeauthoryear{Bard}{Bard}{2009}]{monograph}
{\sc Bard, G.~V.} 2009.
\newblock {\em Algebraic Cryptanalysis}.
\newblock Springer Verlag, Berlin, Heidelberg, New York.

\bibitem[\protect\citeauthoryear{Bosma, Cannon, and Playoust}{Bosma
  et~al\mbox{.}}{1997}]{magma}
{\sc Bosma, W.}, {\sc Cannon, J.}, {\sc and} {\sc Playoust, C.} 1997.
\newblock The {MAGMA} {A}lgebra {S}ystem {I}: {T}he {U}ser {L}anguage.
\newblock In {\em Journal of Symbolic Computation 24}. Academic Press,
  235--265.

\bibitem[\protect\citeauthoryear{Coppersmith and Winograd}{Coppersmith and
  Winograd}{1990}]{CW90}
{\sc Coppersmith, D.} {\sc and} {\sc Winograd, S.} 1990.
\newblock Matrix multiplication via arithmetic progressions.
\newblock {\em Journal of Symbolic Computation\/}~{\em 9,\/}~3, 251 -- 280.
\newblock <ce:title>Computational algebraic complexity editorial</ce:title>.

\bibitem[\protect\citeauthoryear{Courtois, Klimov, Patarin, and
  Shamir}{Courtois
  et~al\mbox{.}}{2000}]{courtois-klimov-patarin-shamir:eurocrypt2000}
{\sc Courtois, N.~T.}, {\sc Klimov, A.}, {\sc Patarin, J.}, {\sc and} {\sc
  Shamir, A.} 2000.
\newblock Efficient algorithms for solving overdefined systems of multivariate
  polynomial equations.
\newblock In {\em Advances in Cryptology --- {EUROCRYPT} 2000}. Lecture Notes
  in Computer Science, vol. 1807. Springer Verlag, Berlin, Heidelberg, New
  York, 392--407.

\bibitem[\protect\citeauthoryear{Faug{\`e}re}{Faug{\`e}re}{1999}]{f4}
{\sc Faug{\`e}re, J.-C.} 1999.
\newblock A new efficient algorithm for computing {G}r{\"o}bner basis ({F4}).
\newblock {\em Journal of Pure and Applied Algebra\/}~{\em 139,\/}~1-3, 61--88.

\bibitem[\protect\citeauthoryear{Golub and Van~Loan}{Golub and
  Van~Loan}{1996}]{GoVa96}
{\sc Golub, G.} {\sc and} {\sc Van~Loan, C.} 1996.
\newblock {\em Matrix Computations\/}, third ed.
\newblock {The Johns Hopkins University Press}.

\bibitem[\protect\citeauthoryear{Ibarra, Moran, and Hui}{Ibarra
  et~al\mbox{.}}{1982}]{IbMoHu82}
{\sc Ibarra, O.}, {\sc Moran, S.}, {\sc and} {\sc Hui, R.} 1982.
\newblock A generalization of the fast {LUP} matrix decomposition algorithm and
  applications.
\newblock {\em Journal of Algorithms\/}~{\em 3}, 45--56.

\bibitem[\protect\citeauthoryear{Jeannerod, Pernet, and Storjohann}{Jeannerod
  et~al\mbox{.}}{2011}]{jeannerod-pernet-storjohann:ple2010}
{\sc Jeannerod, C.-P.}, {\sc Pernet, C.}, {\sc and} {\sc Storjohann, A.} 2011.
\newblock Rank profile revealing {G}aussain elimination and the {CUP} matrix
  decomposition.
\newblock {\em pre-print to appear\/}.

\bibitem[\protect\citeauthoryear{Lazard}{Lazard}{1983}]{lazard:eurocal83}
{\sc Lazard, D.} 1983.
\newblock Gr\"obner-bases, {G}aussian elimination and resolution of systems of
  algebraic equations.
\newblock In {\em Proceedings of the European Computer Algebra Conference on
  Computer Algebra}. Lecture Notes in Computer Science, vol. 162. Springer
  Verlag, Berlin, Heidelberg, New York.

\end{thebibliography}

\appendix
\section{The M4RI Algorithm} \label{app:m4ri}
Consider the matrix $A$ of dimension $m \times n$ in Figure~\ref{fig:m4ri-visualisation}. Because $k=3$ here, the $3 \times n$ submatrix on the top has full rank and we performed Gaussian elimination on it. Now, we need to clear the first $k$ columns of $A$ in the rows below $k$ (and above the submatrix in general if we want the reduced row echelon form). There are $2^k$ possible linear combinations of the first $k$ rows, which we store in a table $T$. We index $T$ by the first $k$ bits (e.g. $0 1 1 \rightarrow 3$). Now to clear $k$ columns of row $i$ we use the first $k$ bits of that row as an index in $T$ and add the matching row of $T$ to row $i$, causing a cancellation of $k$ entries. Instead of up to $k$ additions, this only costs one addition, due to the pre-computation. Using Gray codes (see also \cite{matmulgf2}) or similar techniques this pre-computation can be performed in $2^k$ vector additions and the overall cost is $2^k + m - k + k^2$ vector additions in the worst case (where $k^2$ accounts for the Gauss elimination of the $k \times n$ submatrix). The naive approach would cost $k \cdot m$ row additions in the worst case to clear $k$ columns. If we set $k = \log m$ then the complexity of clearing $k$ columns is $\ord{m + \log^2 m}$ vector additions in contrast to $\ord{m \cdot \log m}$ vector additions using the naive approach.

\begin{figure}[htbp]
\begin{align*}
A = \left(\begin{array}{rrr|rrrrrr}
1 & 0 & 0 & 1 & 0 & 1 & 1 & 1 & \dots \\
0 & 1 & 0 & 1 & 1 & 1 & 1 & 0 & \dots \\
0 & 0 & 1 & 0 & 0 & 1 & 1 & 1 & \dots \\
\dots \\
0 & 0 & 0 & 1 & 1 & 0 & 1 & 0 & \dots \\
\bf{1} & \bf{1} & \bf{0} & 0 & 1 & 0 & 1 & 1 & \dots \\
0 & 1 & 0 & 0 & 1 & 0 & 0 & 1 & \dots \\
\dots \\
\bf{1} & \bf{1} & \bf{0} & 1 & 1 & 1 & 0 & 1 & \dots 
\end{array}\right)
T = \left[\begin{array}{rrr|rrrrrr}
0 & 0 & 0 & 0 & 0 & 0 & 0 & 0 & \dots\\
0 & 0 & 1 & 0 & 0 & 1 & 1 & 1 & \dots \\
0 & 1 & 0 & 1 & 1 & 1 & 1 & 0 & \dots \\
0 & 1 & 1 & 1 & 1 & 0 & 0 & 1 & \dots \\
1 & 0 & 0 & 1 & 0 & 1 & 1 & 1 & \dots \\
1 & 0 & 1 & 1 & 0 & 0 & 0 & 0 & \dots \\
\bf{1} & \bf{1} & \bf{0} & 0 & 1 & 0 & 0 & 1 & \dots \\
1 & 1 & 1 & 0 & 1 & 1 & 1 & 0 & \dots \\
\end{array}\right]
\end{align*}
\caption{M4RI Idea}
\label{fig:m4ri-visualisation}
\end{figure}

This idea leads to Algorithm~\ref{alg:m4ri}. In this algorithm the subroutine \gausssubmatrix\ (see Algorithm~\ref{alg:gaussubmatrix}) performs Gauss elimination on a $k \times n$ submatrix of $A$ starting at position $(r,c)$ and searches for pivot rows up to $m$. If it cannot find a submatrix of rank $k$ it will terminate and return the rank $\kbar$ found so far.

The subroutine \maketable\ (see Algorithm~\ref{alg:maketable}) constructs the table $T$ of all $2^k$ linear combinations of the $k$ rows starting at row $r$ and column $c$. More precisely, it enumerates all elements of the vector space $\mathsf{span}(r,...,r+\kbar-1)$ spanned by the rows $r,\dots, r+\kbar-1$. Finally, the subroutine \addrowsfromtable\ (see Algorithm~\ref{alg:addrows}) adds the appropriate row from $T$ --- indexed by $k$ bits starting at column $c$ --- to each row of $A$ with index $i \not\in \{r,\dots,r+\kbar-1\}$. That is, it adds the  appropriate linear combination of the rows $\{r,\dots,r+\kbar-1\}$  onto a row $i$ in order to clear $k$ columns.

\begin{algorithm}
\KwIn{$A$ -- an $m \times n$ matrix}
\KwIn{$r_{\textnormal{start}}$ -- an integer 
$0 \leq r_{\textnormal{start}} < m$}
\KwIn{$r_{\textnormal{end}}$ -- an integer 
$0 \leq r_{\textnormal{start}} \leq r_{\textnormal{end}} < m$}
\KwIn{$c_{\textnormal{start}}$ -- an integer
 $0 \leq c_{\textnormal{start}} < n$}
\KwIn{$k$ -- an integer $k > 0$}
\KwIn{$T$ -- a $2^k \times n$ matrix}
\KwIn{$L$ -- an integer array of length $2^k$}
\Begin{
\For{$r_{\textnormal{start}} \leq i < r_{\textnormal{end}}$}{
  $id = \sum_{j=0}^{k} A_{i,c_\textnormal{start} + j} \cdot 2^{k-j-1}$\;
  $j \longleftarrow L_{id}$\;
  add row $j$ from $T$ to the row $i$ of $A$ starting at
column $c_\textnormal{start}$\;
} 
}
\caption{\textsc{AddRowsFromTable}}
\label{alg:addrows}
\end{algorithm}

\begin{algorithm}
\KwIn{$A$ -- an $m \times n$ matrix}
\KwIn{$r_{\textnormal{start}}$ -- an integer 
$0 \leq r_{\textnormal{start}} < m$}
\KwIn{$c_{\textnormal{start}}$ -- an integer
 $0 \leq c_{\textnormal{start}} < n$}
\KwIn{$k$ -- an integer $k > 0$}
\KwResult{Retuns an $2^k \times n$ matrix $T$}
\Begin{
$T \longleftarrow$ the $2^k \times n$ zero matrix\;
\For{$1 \leq i < 2^k$}{
  $j \longleftarrow$ the row index of $A$ to add according to the
Gray code\;
 add row $j$ of $A$ to the row $i$ of $T$ starting at
column $c_\textnormal{start}$\;
}
 $L \longleftarrow$ integer array allowing to index $T$ by $k$ bits
starting at column $c_{\textnormal{start}}$\;
\Return{$T,L$}\;
}
\caption{\textsc{MakeTable}}
\label{alg:maketable}
\end{algorithm}

\begin{algorithm}
\KwIn{$A$ -- an $m \times n$ matrix}
\KwIn{$r$ -- an integer $0 \leq r < m$}
\KwIn{$c$ -- an integer $0 \leq c < n$}
\KwIn{$k$ -- an integer $k > 0$}
\KwIn{$r_\textnormal{end}$ -- an integer $0 \leq r \leq r_\textnormal{end} <
m$}
\KwResult{Returns the rank $\kbar \leq k$ and puts the $\kbar
\times (n-c)$ submatrix starting at $A[r,c]$ in reduced row echelon form.}
\Begin{
  $r_{start} \longleftarrow r$\;
  \For{$c \leq j < c + k$}{
    $found \longleftarrow \texttt{False}$\;
    \For{$r_{start} \leq i < r_\textnormal{end}$}{
       \For(\tcp*[h]{clear the first columns}){$0 \leq l < j-c$}{
          \lIf{$A[i,c+l] \neq 0$}{
            add row $r+l$ to row $i$ of $A$ starting at column $c+l$\;
          }
       }
       \If(\tcp*[h]{pivot?}){$A[i,j] \neq 0$}{
         Swap the rows $i$ and $r_{start}$ in $A$\;
         \For(\tcp*[h]{clear above}){$r \leq l < r_{start}$}{
           \lIf{$A[l,j]\neq 0$}{
             add row $r_{start}$ to row $l$ in $A$ starting at column $j$\;
           }
         }
         $r_{start} \longleftarrow r_{start} + 1$\;
         $found \longleftarrow True$\;
         break\;
       }
    }
    \If{$found = $\texttt{False}}{
      \Return{j - c}\;
    }
  }
  \Return{j - c}\;
}
\caption{\textsc{GaussSubmatrix}}
\label{alg:gaussubmatrix}
\end{algorithm}

\begin{algorithm}
\KwIn{$A$ -- an $m \times n$ matrix}
\KwIn{$k$ -- an integer $k > 0$}
\KwResult{$A$ is in reduced row echelon form.}
\Begin{
$r,c \longleftarrow 0,0$\;
\While{$c<n$}{
   \lIf{$c+k > n$}{
     $k \leftarrow n - c$\;
   }
   $\kbar \longleftarrow$ \textsc{GaussSubmatrix}($A, r, c, k, m$)\;
   \If{$\kbar > 0$}{
     $T,L \longleftarrow$ \textsc{MakeTable}($A, r, c, \kbar$)\;
     \addrowsfromtable($A, 0, r, c, \kbar, T, L$)\;
     \addrowsfromtable($A, r+\kbar, m, c, \kbar, T, L$)\;
   }
   $r,c \longleftarrow r + \kbar, c + \kbar$\;
   \lIf{$k \neq \kbar$}{
     $c \leftarrow c + 1$\;
   }
}
}
\caption{M4RI}
\label{alg:m4ri}
\end{algorithm}

When studying the performance of Algorithm~\ref{alg:m4ri}, we expect the function \maketable\ to contribute most. Instead of performing $\kbar/2 \cdot 2^{\kbar} - 1$ additions \maketable\ only performs $2^{\kbar}-1$ vector additions. However, in practice the fact that $\kbar$ columns are processed in each loop iteration of \addrowsfromtable\ contributes signficiantly due to the better cache locality. Assume the input matrix $A$ does not fit into L2 cache. Gaussian elimination would load a row from memory, clear one column and likely evict that row from cache in order to make room for the next few rows before considering it again for the next column. In the M4RI algorithm more columns are cleared per memory load instruction.

We note that our presentation of M4RI differs somewhat from that in \cite{Ba07}. The key difference is that our variant does not throw an error if it cannot find a pivot within the first $3k$ rows in \gausssubmatrix. Instead, our variant searches all rows and consequently the worst-case complexity is cubic. However, on average for random matrices we expect to find a pivot within $3k$ rows with overwhelming probability \cite[Ch. 9.5.3]{monograph} and thus the average-case complexity is $\ord{n^3/\log n}$.

In the Table~\ref{tab:m4ri}, we give running times for computing the reduced row echelon form (or row echelon form for NTL) for random matrices of various dimensions. The column `M4RI' in both tables refers to our implementation of Algorithm~\ref{alg:m4ri} using implementation techniques from \cite{matmulgf2} such as multiple pre-computation tables. The implementation
that NTL uses is straight-forward Gaussian elimination and thus serves as a baseline for comparison while \Magma implements asymptotically fast matrix elimination.

\begin{table}[htbp]
\begin{footnotesize}
\begin{center}
\begin{tabular}{|c|r|r|r|}
\hline
 \multicolumn{4}{|c|}{64-bit Debian/GNU Linux, 2.33Ghz \CTD}\\
\hline
Matrix & {\Magma} & {NTL} & {M4RI}\\
Dimensions & {2.14-17} & 5.4.2 & 20090105\\
\hline
$10,000 \times 10,000$ &   2.474s &   14.140s &   1.532s\\
$16,384 \times 16,384$ &   7.597s &   67.520s &   6.597s\\
$20,000 \times 20,000$ &  13.151s &  123.700s &  12.031s\\
$32,000 \times 32,000$ &  39.653s &  462.320s &  40.768s\\
$64,000 \times 64,000$ & 251.346s & 2511.540s & 241.017s\\
\hline
\hline
 \multicolumn{4}{|c|}{64-bit Debian/GNU Linux, 2.6Ghz \Opteron}\\
\hline
Matrix & {\Magma} & {NTL} & {M4RI}\\
Dimensions & {2.14-13} & 5.4.2 & 20090105\\
\hline
$10,000 \times 10,000$ &   3.351s &   18.45s &   2.430s\\
$16,384 \times 16,384$ &  11.289s &   72.89s &  10.822s\\
$20,000 \times 20,000$ &  16.734s &  130.46s &  19.978s\\
$32,000 \times 32,000$ &  57.567s &  479.07s &  83.575s\\
$64,000 \times 64,000$ & 373.906s & 2747.41s & 537.900s\\
\hline
\end{tabular}
\caption{RREF for Random Matrices}
\label{tab:m4ri}
\end{center}
\end{footnotesize}
\end{table}

Table~\ref{tab:m4ri} shows that our implementation of the M4RI algorithm (with average-case complexity $\ord{n^3/\log n}$) is competitive with \Magma's asymptotically fast algorithm (with complexity $\ord{n^\omega}$) up to $64,000 \times 64,000$ on the \CTD and up to $16,384 \times 16,384$ on the \Opteron.
\end{document}